\documentclass{article}
\usepackage[english]{babel}
\usepackage{graphicx,multicol}
\usepackage{epic,eepic,epsfig}
\usepackage{amssymb}
\usepackage{amsmath,amsthm}
\usepackage{color}
\usepackage{tikz}
\usepackage{epsfig,url}
\usepackage{multirow}

\usetikzlibrary{shapes,snakes}

\usepackage{eso-pic} 

\newcommand{\emphX}[1]{\textbf{#1}}

\newcommand{\induce}[2]{\mbox{$ #1 \langle #2 \rangle$}}

\newcommand{\dom}{\mbox{$\rightarrow$}}

\newcommand{\be}{\begin{enumerate}}
\newcommand{\ee}{\end{enumerate}}
\newcommand{\bd}{\begin{description}}
\newcommand{\ed}{\end{description}}

\newcommand{\beq}{\begin{equation}}
\newcommand{\eeq}{\end{equation}}

\newcommand{\2}{\vspace{2mm}}

\renewenvironment{proof}[1][]{\par \noindent {\bf Proof#1}.\ }{\hfill$\Box$
\par \vspace{11pt}}
\newenvironment{subproof}{\par\noindent {\it Subproof}.\ }{\hfill$\lozenge$\par\vspace{11pt}}

\newtheorem{theorem}{Theorem}[section]
\newtheorem{lemma}[theorem]{Lemma}
\newtheorem{proposition}[theorem]{Proposition}

\newtheorem{claim}{Claim}[theorem]

\theoremstyle{definition}

\newtheorem{problem}[theorem]{Problem}

\begin{document}
\bibliographystyle{plain}

\title{Out-degree reducing partitions of digraphs}
\author{J. Bang-Jensen$^1$, S. Bessy$^2$, F. Havet$^3$ and A. Yeo$^{1,4}$}
\maketitle

\begin{center}
{\small $^1$ Department of Mathematics and Computer Science, University of Southern 
Denmark, Odense DK-5230, Denmark, Email: \verb!{jbj,yeo}@imada.sdu.dk!. Financial support: Danish research council, grant number 1323-00178B and Labex UCN@Sophia\\
$^2$ LIRMM, Universit\'e de Montpellier, Montpellier, France.
Email: \verb!stephane.bessy@lirmm.fr!. \\Financial support: OSMO project, 
 Occitanie regional council.\\
$^3$ Universit\'e C\^ote d'Azur, CNRS, I3S  and INRIA, Sophia Antipolis, France.
Email: \verb!frederic.havet@cnrs.fr!. \\
Financial support: ANR-13-BS02-0007 STINT.\\
$^4$ Department of Mathematics,  University of Johannesburg,  Auckland Park, 2006 South Africa. }
\end{center}

\begin{abstract}
Let $k$ be a fixed integer. We determine the complexity of finding a
$p$-partition $(V_1, \dots, V_p)$ of the vertex set of a given digraph such
that the maximum out-degree of each of the digraphs induced by $V_i$, ($1\leq i\leq p$)
is at least $k$ smaller than the maximum
out-degree of $D$. We show that this problem is polynomial-time
solvable when $p\geq 2k$ and ${\cal NP}$-complete otherwise. The result for
$k=1$ and $p=2$ answers a question posed in \cite{bangTCS636}.  We also
determine, for all fixed non-negative integers $k_1,k_2,p$, the
complexity of deciding whether a given digraph of maximum
  out-degree $p$ has a $2$-partition $(V_1,V_2)$ such that the
digraph induced by $V_i$ has maximum out-degree at most $k_i$ for
$i\in [2]$. It follows from this characterization that the problem of
deciding whether a digraph has a 2-partition $(V_1,V_2)$ such that
each vertex $v\in V_i$ has at least as many neighbours in the set
$V_{3-i}$ as in $V_i$, for $i=1,2$ is ${\cal NP}$-complete. This
solves a problem from \cite{kreutzerEJC24} on majority colourings.  

\medskip

\noindent{}{\bf Keywords:} $2$-partition, maximum out-degree reducing
partition, ${\cal NP}$-complete, polynomial algorithm.
\end{abstract}

\section{Introduction}

A {\bf $p$-partition} of a graph or digraph $G$ is a vertex partition
$(V_1,V_2,\ldots{},V_p)$ of its vertex set $V(G)$. 

It is a well-known and easy fact that every undirected graph $G$
admits $2$-partition such that the degree of each vertex in its part
is at most half of its degree in $G$ and such a partition can be found
by a greedy algorithm (or by considering a maximum-cut partition). So
we have the following.

\begin{proposition}~\
\label{GreduceDelta}
\begin{itemize}
\item[(i)] Every graph $G$ has a $2$-partition $(V_1,V_2)$ such that
  $d_{G\langle V_i\rangle}(v) \leq d_G(v)/2$ for all $i\in\{1,2\}$ and
  all $v\in V_i$.
\item[(ii)] Every graph $G$ has a $2$-partition $(V_1,V_2)$ with
  $\Delta(\induce{G}{V_i})\leq \Delta(G)/2$ for $i=1,2$.
\end{itemize}
\end{proposition}

Thomassen \cite{thomassenEJC6} constructed an infinite class of
strongly connected digraphs ${\cal T}=T_1,T_2,\ldots{},T_k,\ldots{}$
with the property that for each $k$, $T_k$ is $k$-out-regular and has
no even directed cycle. As remarked by Alon in \cite{alonCPC15} this
implies that we cannot expect any directed analogues of the statements
in Proposition~\ref{GreduceDelta}.

\begin{proposition}
\label{monochromatic}
Let $k$ be a positive integer. For every $2$-partition $(V_1,V_2)$ of $T_k$, some
 vertex has all its $k$ out-neighbours in the same part as
 itself, so $\max \{\Delta^+(\induce{D}{V_1}),
 \Delta^+(\induce{D}{V_2})\} = \Delta^+(D)$.
 \end{proposition}

This is due to the simple fact that if a digraph $D$ has a
$2$-partition $(V_1,V_2)$ such that the bipartite digraph induced by
the arcs between the two sets has minimum out-degree at least $1$,
then this and hence also $D$ has an even directed cycle.

Alon~\cite{alonCPC15} also remarked that it is always possible to
split $V(D)$ into three sets such that each of the induced subdigraphs
has smaller maximum out-degree than $D$ (see Theorem
\ref{alonsplit}). In Proposition~\ref{prop:Alon-gen}, we generalize
this to all values of $k$.  We show that for every positive integer
$k$, there is a $(2k+1)$-partition of $V(D)$ such that the out-degree
of every vertex $x$ in its part is at most $d^+_D(x)-k$ or $0$ if
$d^+_D(x) <k$.

\medskip

The digraphs in ${\cal T}$ show that one cannot always obtain a
$2$-partition of a digraph such that in each subdigraph induced by the
parts, the out-degree of every vertex or the maximum out-degree is
smaller than in the original graph. So it is natural to ask whether
the existence of such a partition can be decided in polynomial time.

A {\bf $k$-all-out-degree-reducing $p$-partition} of a digraph $D$ is
a $p$-partition $(V_1,\dots , V_p)$ of $V$ such that
$d^+_{\induce{D}{V_i}}(v) \leq \max\{0, d^+_D(v)-k\}$ for all $1\leq
i\leq p$ and all $v\in V_i$.  A {\bf $k$-max-out-degree-reducing
  $p$-partition} of a digraph $D$ is a $p$-partition $(V_1,\dots ,
V_p)$ of $V$ such that $\Delta^+(\induce{D}{V_i})\leq
  \max\{0,\Delta^+(D)-k\}$ for $i\in [p]$.  Observe that a
  $k$-all-out-degree-reducing $p$-partition is also a
  $k$-max-out-degree-reducing $p$-partition.  However the converse is
  not necessarily true.  So for fixed integers $k$ and $p$, we are
interested in the problems of deciding whether a given digraph admits
one of the above defined partitions.

\begin{problem}[\sc $k$-all-out-degree-reducing $p$-partition]
\label{dreducebyk}~\\
\underline{Input:} a digraph $D$;\\ \underline{Question:} Does $D$
have a $p$-partition $(V_1,V_2)$ with $d^+_{\induce{D}{V_i}}(v) \leq
\max\{0, d^+_D(v)-k\}$ for $i\in [p]$?
\end{problem}

\begin{problem}[\sc $k$-max-out-degree-reducing $p$-partition]
\label{Deltareducebyk}~\\
\underline{Input:} a digraph $D$;\\ \underline{Question:} Does $D$
have a $p$-partition $(V_1,V_2)$ with
  $\Delta^+(\induce{D}{V_i})\leq \max\{0,\Delta^+(D)-k\}$ for $i\in [p]$?
\end{problem}

\medskip

We first consider the case of $2$-partitions. The complexity of {\sc $1$-max-out-degree-reducing $2$-partition} was posed in the paper~\cite{bangTCS636} in which  the complexity of a large number of other $2$-partition problems is established.
We also consider a closely related kind of $2$-partitions: A {\bf
  $(\Delta^+\leq k_1, \Delta^+\leq k_2)$-partition} of a digraph is a
$2$-partition $(V_1, V_2)$ such that $\Delta^+(\induce{D}{V_i})\leq
k_i$ for $i\in \{1,2\}$.  Note that if a digraph is $r$-out-regular,
then a $(\Delta^+\leq r-k, \Delta^+\leq r-k)$-partition is also a
$k$-max-out-degree-reducing $2$-partition and a
$k$-all-out-degree-reducing $2$-partition.  We thus consider the
following problem.
\begin{problem}[\sc $(\Delta^+\leq k_1, \Delta^+\leq k_2)$-Partition]
\label{Deltak1k2}~\\
\underline{Input:} a  digraph $D$;\\
\underline{Question:} Does $D$  have a $2$-partition $(V_1,V_2)$
with $\Delta^+(\induce{D}{V_i})\leq k_i$ for $i\in \{1,2\}$?
\end{problem}

When $k_1=k_2=0$ the problem is the same as just asking whether $D$ is
bipartite which is clearly polynomial-time solvable.  If $D$ is a
symmetric digraph, then there is a one-to-one correspondence between
the set of $(\Delta^+\leq k, \Delta^+\leq k)$-partitions of $D$ and
the so-called $k$-improper $2$-colourings of $UG(D)$, the underlying
(undirected) graph of $D$. A 2-colouring is {\bf $k$-improper} if no
vertex has more than $k$ neighbours with the same colour as itself.
Cowen et al.~\cite{cowenLNCS713} proved that for any $k\geq 1$,
deciding whether a graph has a $k$-improper $2$-colouring is ${\cal
  NP}$-complete. Consequently, {\sc $(\Delta^+\leq k, \Delta^+\leq
  k)$-Partition} is ${\cal NP}$-complete for all $k\geq 1$.

On the other hand, Proposition \ref{GreduceDelta} (ii) can be
translated as follows to symmetric digraphs.
\begin{proposition}
Every symmetric digraph with maximum out-degree $K$ has a
$(\Delta^+\leq \lfloor K/2 \rfloor, \Delta^+\leq \lfloor K/2
\rfloor)$-partition.
\end{proposition}

As we saw in Proposition \ref{monochromatic}, this result does not
extend to general digraphs.  Hence it is natural to ask about the
complexity of {\sc $(\Delta^+\leq k_1, \Delta^+\leq k_2)$-Partition}
when restricted to digraphs with small maximal out-degree.

\medskip

In the first part of the paper, we prove that {\sc $1$-all-out-degree-reducing $2$-partition} and {\sc $1$-max-out-degree-reducing $2$-partition} can be solved in polynomial
time when $k=1$. This answers the question posed in \cite{bangTCS636}
affirmatively. Then we derive a complete characterization of the
complexity of Problem \ref{Deltak1k2} in terms of the values of
$k_1,k_2$ and use it to prove that  {\sc $k$-all-out-degree-reducing $2$-partition}
and~{\sc $k$-max-out-degree-reducing $2$-partition} are ${\cal NP}$-complete for all values of
$k$ higher than 1. As a consequence of these results we solve an open
problem from \cite{kreutzerEJC24} on majority colourings.

Next, in Section\ref{sec:p-part}, we consider $p$-partitions for $p\geq 3$. We show that every digraph that a $k$-all-out-degree-reducing $2k+1$-partition. This implies that {\sc $k$-all-out-degree-reducing $p$-partition} and {\sc $k$-max-out-degree-reducing $p$-partition} are polynomial-time solvable for $p\geq 2k+1$ as the answer is always `Yes'.
We also characterize the digraphs having a $k$-all-out-degree-reducing $2k+1$-partition, which implies that {\sc $k$-all-out-degree-reducing $2k$-partition} and {\sc $k$-max-out-degree-reducing $2k$-partition} are polynomial-time solvable.
Finally, we show that,  for any $k>1$ and $3\leq p\leq 2k-1$, the problems {\sc $k$-all-out-degree-reducing $p$-partition} and {\sc $k$-max-out-degree-reducing $k$-partition} are ${\cal NP}$-complete.

We  conclude with some remarks and related open problems.

\section{Notation}

Notation generally follows \cite{bang2009}.  We use the shorthand
notation $[k]$ for the set $\{1,2,\ldots{},k\}$.  Let $D=(V,A)$ be a
digraph with vertex set $V$ and arc set $A$.

Given an arc $uv\in A$, we say that $u$ \emphX{dominates} $v$ and $v$
is \emphX{dominated} by $u$. If $uv$ or $vu$ (or both) are arcs of
$D$, then $u$ and $v$ are {\bf adjacent}.  If neither $uv$ or $vu$
exist in $D$, then $u$ and $v$ are {\bf non-adjacent}.  The {\bf
  underlying graph} of a digraph $D$, denoted by $UG(D)$, is obtained
from $D$ by suppressing the orientation of each arc and deleting
multiple copies of the same edge (coming from directed 2-cycles).  A
digraph $D$ is {\bf connected} if $UG(D)$ is a connected graph, and
the {\bf connected components} of $D$ are those of $UG(D)$.

A {\bf $(u,v)$-path} is a directed path from $u$ to $v$.  A digraph is {\bf
  strongly connected} (or {\bf strong}) if it contains a $(u,v)$-path
for every ordered pair of distinct vertices $u,v$.  A digraph $D$ is
{\bf $k$-strong} if for every set $S$ of less than $k$ vertices the
digraph $D-S$ is strong.  A {\bf strong component} of a digraph $D$ is
a maximal subdigraph of $D$ which is strong. A strong component is
{\bf trivial}, if it has order $1$. An {\bf initial} (resp. {\bf
  terminal}) strong component of $D$ is a strong component $X$ with no
arcs entering (resp. leaving) $X$ in $D$.

The \emphX{subdigraph induced} by a set of vertices $X$ in a digraph
$D$, denoted by \induce{D}{X}, is the digraph with vertex set $X$ and
which contains those arcs from $D$ that have both end-vertices in
$X$. When $X$ is a subset of the vertices of $D$, we denote by $D-X$
the subdigraph $\induce{D}{V\setminus X}$. If $D'$ is a subdigraph of
$D$, for convenience we abbreviate $D-V(D')$ to $D-D'$.  

The \emphX{in-degree} (resp. \emphX{out-degree}) of $v$, denoted by
$d^-_D(v)$ (resp. $d^+_D(v)$), is the number of arcs from $V\setminus
\{v\}$ to $v$ (resp. $v$ to $V\setminus \{v\}$).  A digraph is {\bf
  $k$-out-regular} if all its vertices have out-degree $k$
and it is {\bf $k$-regular} if every vertex has both
in-degree and out-degree $k$.  A {\bf sink} is a vertex with
out-degree $0$ and a {\bf source} is a vertex with in-degree $0$.  The
{\bf degree} of $v$, denoted by $d_D(v)$, is given by
$d_D(v)=d^+_D(v)+d^-_D(v)$.  Finally the \emphX{maximum
    out-degree} and \emphX{maximum in-degree} of $D$ are respectively
  denoted by $\Delta^+(D)$ and $\Delta^-(D)$.

An {\bf out-tree} rooted at the vertex $s$, also called an {\bf
  $s$-out-tree},
is a connected digraph $T^+_s$ such that $d^-_{T^+_s}(s) =0$ and
$d^-_{T^+_s}(v)=1$ for every vertex $v$ different from $s$.
Equivalently, for every $v\in V(T^+_s)\setminus \{s\}$ there is a
unique $(s,v)$-path in $T^+_s$.  

In our ${\cal NP}$-completeness proofs we use reductions from the
well-known 3-SAT problem and from {\sc Monotone
  Not-all-equal-3-SAT}. The later is the variant where the boolean
formula ${\cal F}$ to be satisfied consists of clauses all of whose
literals are non-negated variables and we seek a truth assignment such
that each clause will get both a true and a false literal.  This
problem is also ${\cal NP}$-complete
\cite{shaeferSTOC10}.

\section{1-out-degree reducing partitions of digraphs}\label{outdegredsec}

In this section we prove that {\sc $1$-all-out-degree-reducing $2$-partition}
  and~{\sc $1$-max-out-degree-reducing $2$-partition} are solvable in polynomial time for $k=1$.
 
Part (i) of the theorem below follows from a result of Seymour
\cite{seymourQJM25} (see also \cite{kreutzerEJC24}) but we include the
short proof for completeness (and we use the same idea to prove
(ii)). We shall use the following result, due to Robertson, Seymour,
and Thomas.

\begin{theorem}[Robertson, Seymour, and Thomas \cite{robertsonAM150}]\label{thm:even-cycle}
Deciding whether a given digraph has an even directed cycle is
polynomial-time solvable.
\end{theorem}

\begin{theorem}\label{thm:struc-k1-p1}
Let $D$ be a digraph.
\begin{itemize}
\item[(i)] $D$ admits a $1$-all-out-degree-reducing-$2$-partition if
  and only if every non-trivial terminal strong component contains an
  even directed cycle.
\item[(ii)] $D$ admits a $1$-max-out-degree-reducing-$2$-partition if
  and only if every terminal strong component contains an even
  directed cycle or a vertex with out-degree less than $\Delta^+(D)$.
\end{itemize}
\noindent{}In both cases above, the desired $2$-partition can be
constructed in polynomial time when it exists.
\end{theorem}

\begin{proof}
Let $X_1, \dots , X_r$ be the terminal strong components of $D$
ordered in such a way that $X_1,\ldots{},X_q$ are non-trivial and
$X_{q+1},\ldots{},X_r$ are trivial. Set $S=\bigcup_{i=q+1}^r
V(X_i)$. Observe that $S$ is the set of sinks of $D$.

(i) Suppose first that $D$ admits a
$1$-all-out-degree-reducing-$2$-partition, then that partition
restricted to $X_i$, $1\leq i\leq q$, would induce a bipartite
spanning subdigraph of $X_i$ with an even directed cycle.

Assume now that $X_i$ contains an even directed cycle $C_i$ for all
$i\in [q]$.  First properly $2$-colour all the cycles $C_1,C_2, \ldots , C_q$
and colour the vertices of $S$ with colour $1$. If there exists an
uncoloured vertex, then there must also exist an uncoloured vertex
with an arc to a coloured one (as we have coloured at least one vertex
in every terminal strong component).  Give this vertex the opposite
colour of its coloured out-neighbour. Repeating this procedure until
all vertices have been coloured gives us a $2$-colouring where every
vertex not in $S$ has an out-neighbour of different colour to itself.
From this $2$-colouring, we obtain the desired partition.

(ii) The necessity is seen as above. Now assume that every terminal
component $X_i$, $i\in [r]$, contains either an even directed cycle or
a vertex of out-degree less than $\Delta^+(D)$. Pick an even directed
cycle $C_i$ for each terminal component with such a cycle and a vertex
$z_j$ with $d^+(z_j)<\Delta^+(D)$ for the other terminal components
(this includes the trivial ones). Let $Z$ be the union of the vertices
$z_j$.  Now $2$-colour all the even directed cycles and colour the
vertices of $Z$ with colour $1$.  As above we can extend this
colouring into a $2$-colouring of $D$ where every vertex not in $Z$
has an out-neighbour of different colour to itself.  This
$2$-colouring correspond to the desired partition.

\medskip

The complexity claim follows from Theorem~\ref{thm:even-cycle} and the
fact that our proof is constructive. \end{proof}

We will show in Theorem~\ref{thm:recap} that {\sc $k$-all-out-degree-reducing $2$-partition}
and~{\sc $k$-max-out-degree-reducing $2$-partition} are $\cal NP$-complete for $k>1$.

\section{2-partitions with restricted maximum out-degrees}

In this section we consider Problem \ref{Deltak1k2} and determine its
complexity for all possible values of the parameters $k_1,k_2$. By
symmetry, we may assume that we always have $k_1\leq k_2$. Recall that
when $k_1=k_2=0$ the problem is the same as just asking whether $D$ is
bipartite which is polynomial-time solvable, so we may assume below
that $k_2>0$.

\medskip

The following gadget, depicted in Figure~\ref{figureConnector}, turns
out to be very useful in our constructions.  An {\bf
  $(x,y)$-$(i,p)$-connector} is the digraph with vertex set
$\{x,y,s\}\cup T\cup U\cup U'$ with $|T|=i$ and $|U|=|U'|=p$ with all
arcs from $x$ to $T$, all arcs from $T$ to $U$, all arcs between $U$
and $U'$, except one arc $u'u$ for some $u\in U$ and $u'\in U'$, all
arcs from $s$ to $U'\setminus \{u'\}$ 
arcs $u's$ and $sy$. The next two lemmas illustrate the usefulness of
connectors.

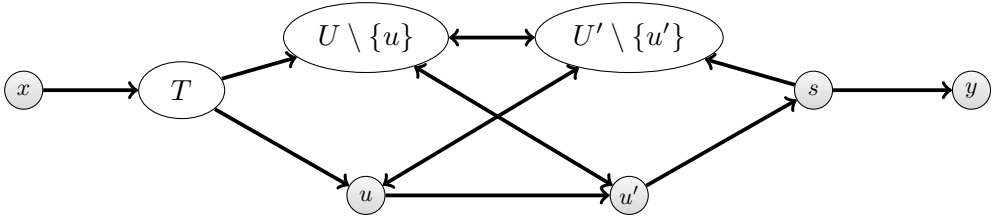
\begin{figure}[h!]
\begin{center}
\tikzstyle{vertexX}=[circle,draw, top color=gray!5, bottom color=gray!30, minimum size=16pt, scale=0.9, inner sep=0.6pt]
\tikzstyle{vertexZ}=[circle,draw, top color=white!5, bottom color=white!30, minimum size=4pt, scale=0.6, inner sep=0.5pt]
\tikzstyle{vertexBIG}=[ellipse, draw, scale=1.1, inner sep=3.5pt]
\begin{tikzpicture}[scale=0.7]

 \node (x) at (1.0,4.0) [vertexX] {$x$};
 \node (T) at (4.0,4.0) [vertexBIG] {\mbox{ }$T$\mbox{ }};
 \node (Ub) at (7.5,5.0) [vertexBIG] {$U\setminus\{u\}$};
 \node (us) at (7.5,2.0) [vertexX] {$u$};
 \node (Um) at (12.5,5.0) [vertexBIG] {$U'\setminus\{u'\}$};
 \node (up) at (12.5,2.0) [vertexX] {$u'$};
 \node (s) at (16.0,4.0) [vertexX] {$s$};
 \node (y) at (19.0,4.0) [vertexX] {$y$};

    \draw [->, line width=0.05cm] (x) -- (T);
    \draw [->, line width=0.05cm] (T) -- (Ub);
    \draw [->, line width=0.05cm] (T) -- (us);


    \draw [<->, line width=0.05cm] (Ub) -- (Um);
    \draw [<->, line width=0.05cm] (Ub) -- (up);
    \draw [<->, line width=0.05cm] (us) -- (Um);
    \draw [->, line width=0.05cm] (us) -- (up);

    \draw [->, line width=0.05cm] (s) -- (Um);
    \draw [->, line width=0.05cm] (up) -- (s);

    \draw [->, line width=0.05cm] (s) -- (y);
\end{tikzpicture}
\end{center}
\caption{An {\bf $(x,y)$-$(i,p)$-connector}, where $|T|=i$ and $|U \setminus \{u\}|=p-1$ and $|U \setminus \{u'\}| = p-1$. }
\label{figureConnector}
\end{figure}

\begin{lemma}\label{lem:connector}
Let $k_1, k_2, i$ be three positive integers, with $1\le k_1\le k_2$,
let $D$ be a digraph and let $x,y$ be two vertices in $D$.  Let $D'$
be the digraph obtained from $D$ by adding an
$(x,y)$-$(i,p)$-connector.  $D'$ has a $(\Delta^+\leq k_1,
\Delta^+\leq k_2)$-partition if and only if $D$ has one.
\end{lemma}
\begin{proof}
Clearly, if $D'$ has a $(\Delta^+\leq k_1, \Delta^+\leq
k_2)$-partition, then its restriction to $V(D)$ is also a
$(\Delta^+\leq k_1, \Delta^+\leq k_2)$-partition.

Assume now that $D$ has a $(\Delta^+\leq k_1, \Delta^+\leq
k_2)$-partition $(V_1, V_2)$.  By symmetry, we may assume that $x\in
V_1$.  Now one easily checks that $(V_1\cup U \cup \{s\}, V_2\cup
T\cup U')$ is a $(\Delta^+\leq k_1, \Delta^+\leq k_2)$-partition of
$D'$. Indeed, even if $y\in V_1$ we have $d_{V_1}^+(s)\leq 1\le k_1$.
\end{proof}

\begin{lemma}\label{lem:regstrong}
Let $k_1,k_2,p$ be non-negative integers with $0\le k_1\le
  k_2$. If {\sc $(\Delta^+\leq k_1, \Delta^+\leq k_2)$-Partition} is
$\cal NP$-complete for digraphs with maximum out-degree $p$, then {\sc
  $(\Delta^+\leq k_1, \Delta^+\leq k_2)$-Partition} is also $\cal
NP$-complete for digraphs for strong $(p+1)$-out-regular digraphs.

\end{lemma}
\begin{proof}
First assume that $1\le k_1\le k_2$. Then we can use
  Lemma~\ref{lem:connector} quite directly. Consider a digraph $D$
with maximum out-degree $p$, and let $\{v_1, \dots , v_n\}$ be its
vertex set.  For $j\in [n]$, let $i_j = p+1 - d^+_D(v_j)$. Observe
that for every $j$ we have $i_j\geq 1$, because $\Delta^+(D)\leq p$.
Let $D'$ be the digraph obtained by adding a $(v_j, v_{j+1})$-$(i_j,p
+1)$-connector for every $j\in [n]$ (with $v_{n+1}=v_1$).  It is
simple matter to check that $D'$ is $(p+1)$-out-regular and strong
because every $i_j$ is at least 1.  Moreover,
Lemma~\ref{lem:connector} implies that $D'$ has a $(\Delta^+\leq k_1,
\Delta^+\leq k_2)$-partition if and only if $D$ has one.

Now assume that we have $0=k_1<k_2$. In this case we will need to
  put connectors between adjacent vertices to insure that
  Lemma~\ref{lem:connector} holds. Indeed if a digraph $D$ has a
  $(\Delta^+=0,\Delta^+\le k_2)$-partition and $xy$ is an arc of $D$,
  then the digraph obtained from $D$ by adding a
  $(x,y)$-$(p+1-d^+_D(x),p)$-connector to $D$ admits also a
  $(\Delta^+=0,\Delta^+\le k_2)$-partition. The proof of this
  statement is similar to the one of Lemma~\ref{lem:connector} using
  the fact that as $xy$ is an arc of $D$ then we cannot have $x\in
  V_1$ and $y\in V_1$ in any $(\Delta^+=0,\Delta^+\le k_2)$-partition
  $(V_1,V_2)$ of $D$.\\ Now let $D$ be a digraph with maximum
  out-degree $p$. It is easy to check that the digraph obtained by
  adding a new vertex to $D$ with two out-neighbours in $D$ has a
  $(\Delta^+=0,\Delta^+\le k_2)$-partition if and only if $D$ has
  one. So let $s$ a new vertex and let $T$ be a binary $s$-out-tree
  with $|V(D)|$ leaves (i.e. every vertex of $T$ has out-degree 2
  except the leaves which have out-degree 0). We construct $D'$ by
  adding a copy of $T$ to $D$ and identifying the vertices of $D$ with
  the leaves of $T$. By repeating the previous remark, we obtain that
  $D'$ admits a $(\Delta^+=0,\Delta^+\le k_2)$-partition if and only
  if $D$ has one. To conclude we build $D''$ by adding a
  $(v,u)$-$(p+1-d^+_{D'}(v),p)$-connector to $D'$ for every arc $uv$
  of the copy of $T$ and a $(s,w)$-$(p-1,p)$-connector for an
  out-neighbour $w$ of $s$. Using the modified version of
  Lemma~\ref{lem:connector} for $(\Delta^+=0,\Delta^+\le
  k_2)$-partitions, we conclude that $D$ has such a partition if, and
  only if, $D$ has one. Moreover, by construction, it is clear that
  $D''$ is strong and $(p+1)$-out-regular.
\end{proof}

Obviously every digraph of maximum out-degree $k\leq \max\{k_1,k_2\}$
has a $(\Delta^+\leq k_1,\Delta^+\leq k_2)$-partition. 
As we now show, just increasing the maximum
out-degree one above this value results in a shift in complexity from
trivial to ${\cal NP}$-complete, even if we also require that the
digraph is strongly connected and out-regular.

\begin{theorem}
\label{k_1<k_2}
For every choice of non-negative integers $k_1,k_2$ with $\max\{1,k_1\}<k_2$, the {\sc $(\Delta^+\leq k_1,\Delta^+\leq k_2)$-Partition}
problem is ${\cal NP}$-complete for strong $(k_2+1)$-out-regular
digraphs.
\end{theorem}

\begin{proof} Let us call a 2-colouring $c: V\dom \{1,2\}$ {\bf good} if the 2-partition
induced by $c$ is a $(\Delta^+\leq k_1,\Delta^+\leq k_2)$-partition.
We start by describing a reduction from 3-SAT to {\sc $(\Delta^+\leq
  k_1,\Delta^+\leq k_2)$-Partition} in graphs of maximum out-degree
$k_2+1$ and then show how to modify the proof to work for strong and
$(k_2+1)$-out-regular digraphs using Lemma~\ref{lem:connector}.

We first make some observations about gadgets that force certain
vertices to have colour 1 or 2 in any good 2-colouring. Let $X$ be the
digraph that we obtain from a copy of the Thomassen digraph $T_{k_2-1}$ (it exists because $k_2>1$)
by adding one new vertex $v$ and all possible arcs from $V(T_{k_2-1})$
to $v$. It follows from Proposition \ref{monochromatic} that in any
good 2-colouring $c$ of a digraph containing an induced copy of $X$
the vertex $v$ must have $c(v)=2$. Let $Z$ be the digraph obtained by
taking $k_2+1$ copies $X_i$, $i\in [k_2+1]$ of $X$, where $v_i$
denotes the copy of $v$ in $X_i$, $i\in [k_2+1]$ and a new vertex $w$
and adding the arcs of $\{v_1v_{1+i} \mid i\in [k_2]\}\cup
\{v_1w\}$. By the remark above, for every good $2$-colouring of a
digraph containing an induced copy of $Z$, we have $c(w)=1$.

When we say below that a certain vertex $u$ has colour 1 or colour 2
we mean that we use a private copy of either  $Z$
with $u=w$  or $X$ with $u=v$to enforce that in all good 2-colourings of $D$ the vertex
$u$ will have the desired colour. Now let $W$ be a digraph containing
$k_1+k_2+2$ vertices
$v,\bar{v},a_1,\ldots{},a_{k_1},b_1\ldots{},b_{k_2}$ and the
arcs of $\{v\bar{v},\bar{v}v\}\cup \{a_1v,a_1\bar{v},b_1v,b_1\bar{v}\}\cup\{a_1a_{j+1} \mid j\in
[k_1-1]\}\cup \{b_1b_{j+1} \mid j\in [k_2-1]\}$. By adding suitable copies
of $X,Z$ we can ensure that for every good colouring of the digraph we
construct below we have $c(a_h)=1$ for $h\in [k_1]$ and $c(b_h)=2$ for
$h\in [k_2]$. This implies that in every good colouring we have
$c(v)=r$ and $c(\bar{v})=3-r$ for some $r\in \{1,2\}$.

Now we are ready to construct a digraph $D=D({\cal F})$ from a given
instance $\cal F$ of 3-SAT. Let ${\cal F}$ have variables
$x_1,x_2,\ldots{},x_n$ and clauses $C_1,C_2,\ldots{},C_m$: represent
each variable $x_i$ by a copy $W_i$ of $W$ where the vertices
$v_i,\bar{v}_i$ correspond to $v$ and $\bar{v}$ in $W$ and play the
role of $x_i,\bar{x}_i$, respectively. For each clause $C_j$, we add a
new vertex $c_j$ of colour 2, $k_2-2$ arcs from $c_j$ to private (to
$c_k$) vertices of colour 2 and three arcs from $c_j$ to the three
vertices that correspond to its literals. So, if
$C_j=(x_1\vee{}\bar{x}_8\vee x_9)$ then we add the arcs
$c_jv_1,c_j\bar{v}_8$ and $c_jv_9$. This completes the construction of
$D$. Clearly $D$ can be constructed in polynomial time given $\cal
F$. The fact that $c_j$ must have colour 2 and already has $k_2-2$
out-neighbours of colour 2 implies that at least one of the vertices
corresponding to the literals of $C_j$ must have colour 1 in any good
colouring. Now it is easy to see that if we associate colour 1 with
$true$, then $D$ has a good colouring if and only if $\cal F$ is
satisfiable. This proves that {\sc $(\Delta^+\leq k_1, \Delta^+\leq
  k_2)$-Partition} is ${\cal NP}$-complete for digraphs of maximum
out-degree $k_2+1$ as it is easy to check that $\Delta^+(D) \leq k_2+1$.

To obtain the result on strong $(k_2+1)$-out-regular digraphs, we
first show how to obtain a strong superdigraph $D'$ of $D$ with the
desired colouring property. First observe that in $D$ no arc enters a
copy of $X$ unless this is inside a copy of $Z$ and for every copy of
$Z$ one copy of $X$ has no arcs entering it.  By adding a new vertex
$s$, sufficiently (but still polynomial in the size of $\cal F$) many
new vertices and the arcs of an out-tree of maximum out-degree $k_2$
rooted at $s$, we can obtain that $s$ is the root of an out-tree
$T^+_s$ whose only intersection with $V(D)$ is in its leaves where
$T^+_s$ has exactly one leaf in each copy of $X$.

 Note that every vertex corresponding to a literal has out-degree 1
 and that every vertex which does not correspond to a literal has a
 directed path to at least one vertex that corresponds to a literal
 (here we use that $T_{k_2-1}$ is strongly connected).  Thus if we add
 the arcs of the directed cycle
 $C=sv_1v_2\ldots{}v_ns$, we obtain the desired strong digraph $D'$ with
 $\Delta^+(D')=k_2+1$.  Clearly $D$ is a subdigraph of $D'$ so every
 good 2-colouring of $D'$ induces a good 2 colouring of
 $D$. Conversely, if $c$ is a good 2-colouring of $V(D)$, then it is
 still a good 2-colouring of $D\cup A(C)$ because $k_2\geq 2$ and we can extend $c$ to the
 non-leaf vertices of $T^+_s$ (colouring them by $2$) because they have out-degree at most
 $k_2$.

It remains to prove that we can also achieve a $(k_2+1)$-out-regular
digraph $D''$ which is strong and has a good 2-colouring if and only
if $\cal F$ is satisfiable. To show this we just have to observe that,
by Lemma \ref{lem:connector}, for every vertex $w$ with out-degree
$k<k_2+1$ we can add a private $(w,w)-(k_2+1-k,k_2+1)$-connector.
\end{proof}

Note that we used the fact that $k_2>1$ at several places in the proof
above. One of these was the use of $T_{k_2-1}$. Hence there still
remains the complexity of {\sc $(\Delta^+\leq 0,\Delta^+\leq
  1)$-Partition}. This was solved by Fraenkel.

\begin{theorem}[Fraenkel~\cite{fraenkelDAM3}]
\label{2direg0-1}
{\sc $(\Delta^+\leq 0,\Delta^+\leq 1)$-Partition} is ${\cal
  NP}$-complete on the class of digraphs with in- and out-degree at
most $2$. 
\end{theorem} 

In order to strengthen this and to unify our results we need the
following result which can be obtained by modifying the proof in
\cite{fraenkelDAM3}. We give a proof for completeness.

\begin{theorem}
\label{0-1part}
For all $p\geq 2$, {\sc $(\Delta^+\leq 0,\Delta^+\leq 1)$-Partition} is
${\cal NP}$-complete on the class of strong $p$-out-regular digraphs.
\end{theorem}

\begin{proof}
By Lemma \ref{lem:regstrong}, it suffices to prove the statement for
$p=2$. A {\bf kernel} in a digraph $D$ is an independent set $K$ of
vertices such that every vertex in $V(G)\setminus K$ has an
out-neighbour in $K$.  Note that $(V_1,V_2)$ is a $(\Delta^+\leq
0,\Delta^+\leq 1)$-partition of the 2-out-regular digraph $D$ if and
only if $V_1$ is a kernel of $D$. We first recall a (slightly simpler
version of) the proof from \cite{fraenkelDAM3} that deciding whether a
digraph has a kernel is ${\cal NP}$-complete for digraphs of maximum
out-degree 2 and then modify that reduction to show that it is $\cal
NP$-complete for strong $2$-out-regular digraphs.

Let $W$ denote the digraph defined by 
$$V(W) = \{z_1,\ldots{},z_9\} \mbox{~~~and~~~}
A(W)=\{z_1z_2,z_2z_3,z_3z_1,z_3z_4,z_4z_5,z_5z_6,z_5z_7,z_6z_8,z_7z_9\}.$$
Now let $\cal F$ be an instance of 3-SAT with variable
$x_1,\ldots{},x_n$ and clauses $C_1,\ldots{},C_m$. Free to duplicate
one clause, we may assume that $m$ is odd.  Form the digraph $G=G(\cal
F)$ by taking one copy $W_j$ of $W$ for each clause $C_j$, $j\in [m]$
(denoting the vertices of $W_j$ by $z_{j,q}$, $q\in [9]$) and adding
$2n$ new vertices $v_1,\bar{v}_1,\ldots{},v_n,\bar{v}_n$, where
$v_i,\bar{v}_i$ correspond to the literals $x_i,\bar{x}_i$ as well as
the arcs $v_i\bar{v}_i,\bar{v}_iv_i$ for $i\in [n]$. Finally, we add
three arcs from each $W_j$ to the vertices that correspond to its
literals so that the vertex $z_{j,8}$ is joined to the vertex
corresponding to the first literal and the vertex $z_{j,9}$ is joined
to the two vertices corresponding to the second and third literal of
$W_j$. Thus if $W_j=(x_4\vee{}x_5\vee{}\bar{x}_8)$, then we add the
arcs $z_{j,8}v_4,z_{j,9}v_5,z_{j,9}\bar{v}_8$. This completes the
construction of $G$. Note that if $K$ is a kernel of $G$, then for
every $j\in [m]$ we have either $\{z_{j,2},z_{j,4},z_{j,6}\}\subset K$
or $\{z_{j,2},z_{j,4},z_{j,7}\}\subset K$ (or both) and this implies
that $|K\cap \{z_{j,8},z_{j,9}\}|\leq 1$. From this it follows that at
least one of the vertices corresponding to the literals of $C_j$ will
belong to $K$. For each $i\in [n]$ we have precisely one of
$v_i,\bar{v}_i$ in $K$ as these vertices are adjacent. Now it is easy
to see that $G$ has a kernel if and only if $\cal F$ is
satisfiable. This shows that deciding whether a digraph has a kernel
and hence {\sc $(\Delta^+\leq 0,\Delta^+\leq 1)$-Partition} is ${\cal
  NP}$-complete for digraphs of maximum out-degree 2.

Let us now prove that it is $\cal NP$-complete for strong
$2$-out-regular digraphs. Note that in $G$ every vertex has out-degree
at least $1$. Let $H$ be the digraph on six vertices $a,b,c,d,e,f$ and
the arcs $de,ef, fd, da, eb, fc, ae,bd,bf,cd,ce$.  Let $G'$ be the
digraph obtained from the disjoint union of $G$ and $H$ and a directed
path $(a_1, \dots , a_m)$ by identifying $a$ and $a_1$ and adding
the arc $a_mz_{m,3}$, the arcs $a_jz_{j,1}$ for $j\in [m]$ and
the arcs $ud$ for every vertex $u$ having out-degree $1$ in $G$.
Clearly, the digraph $G'$ is strong and $2$-out-regular.

Finally let us now prove that $G'$ has a kernel if and only if $G$ has
one. This will immediately imply the result.  If $G$ has a kernel $K$,
then one can easily check that $K\cup \{b,c\}\cup \{a_j \mid
j~\mbox{odd}\}$ is a kernel of $G'$ (recall that $m$ is odd and that
$K$ contains none of $z_{j,1}$, $z_{j,3}$).  Assume now that $G'$ has
a kernel $K'$.  We have $d\notin K'$, for otherwise $b$ and $f$ are
not in $K'$ (because $K'$ is an independent set) and so $e$ has no
out-heighbour in $K'$, a contradiction.  Now all arcs leaving $G$ in
$G'$ have head $d$, so every vertex of $G$ has an out-neighbour in
$K'\cap V(G)$. Hence $K'\cap V(G)$ is a kernel of $G$.
\end{proof}

\begin{theorem}
\label{degrees3}
Let $k,p$ be two positive integers $k$ such that $p\geq k+2$. Then
{\sc $(\Delta^+\leq k,\Delta^+\leq k)$-Partition} is polynomial-time
solvable for digraphs of maximum out-degree $k+1$ and ${\cal
  NP}$-complete on the class of strong $p$-out-regular digraphs.
\end{theorem}
\begin{proof} 
The first part of the claim follows from Theorem
\ref{thm:struc-k1-p1}. Below we show how make a reduction from {\sc
  Monotone Not-All-Equal 3-SAT} to the $(\Delta^+\leq k,\Delta^+\leq
k)$-partition problem in strong $(k+2)$-out-regular
digraphs. Combining this with Lemma \ref{lem:regstrong} proves the
theorem, as $k>0$.

The reduction makes use of the following {\bf forcing gadget}, namely
the digraph $F$ whose vertex set is the union of $X=\{x, x'\}$,
$Y=T_k$ and whose arc set is the union of the arcs of $T_k$ and all
possible arcs from $Y$ to $X$.  The {\bf head} of a forcing gadget is
the set $X$.

\begin{claim}\ 
\begin{itemize}
\item[(i)] In a forcing gadget, all vertices have out-degree $k+2$,
  except those of the head which have out-degree $0$.
\item[(ii)] In any $(\Delta^+ \leq k,\Delta^+ \leq k)$-partition of
  a digraph which contains a copy of the forcing gadget as an induced
  subdigraph, the two vertices of the head are in the same part.
\end{itemize}
\end{claim}
\begin{subproof}
(i) follows from the definition of the forcing gadget as $T_k$ is $k$-out-regular.

(ii) follows from the fact that $Y=T_k$ has no $(\Delta^+\leq k-1,
  \Delta^+\leq k-1)$-partition, implying that in any 2-partition
  $(V_1,V_2)$ of $F$ some vertex of $Y$ already has its $k$
  out-neighbours in $Y$ in the same set $V_i$ as itself and hence both
  $x$ and $x'$ must belong to $V_{3-i}$.
\end{subproof}

Let $\cal F$ be an instance of {\sc Monotone Not-All-Equal $(k+2)$-SAT} on
$n$ variables $x_1, \dots , x_n$ and $m$ clauses $C_1, \dots, C_m$.
For every $i\in [n]$, let $j_1(i) < j_2(i) \cdots < j_{m(i)}(i)$ be
the indices of those clauses in which variable $x_i$ occurs and let
$J(i)=\{j_1(i), \dots, j_{m(i)}(i)\}$.  For each $j\in [m]$ and $q\in
[k+2]$, let $a_{q,j}$ be the unique integer such that if $C_j=x_{i_1}\vee
x_{i_2} \vee x_{i_3}$, then $x_{i_q}$ occurs exactly $a_{q,j}-1$ times
among the clauses $C_1,\ldots{},C_{j-1}$.

Let $D_{{\cal F}}$ be the digraph constructed as follows.  For all
$i\in [n]$, we create a variable gadget $VG_i$ as follows.  We first
create the vertices $\{x_i^{j} \mid j\in J(i)\}$.  Then for all $1\leq
p < m(i)$, we add a forcing gadget with head $\{x_i^{j_p(i)},
x_i^{j_{p+1}(i)}\}$. Let $Y_i^p$ be the set corresponding to $Y$ in
this forcing gadget. This will force all the vertices of $\{x_i^{j}
\mid j\in J(i)\}$ to be in the same part for any $(\Delta^+\leq
k,\Delta^+\leq k)$-partition.

Then for every clause $C_j=x_{i_1}\vee x_{i_2} \vee\ldots\vee
x_{i_{k+2}}$, we add a vertex $t_j$, all the arcs from the set
$\{x_{i_1}^{a_{1,j}}, x_{i_2}^{a_{2,j}},\ldots{},
x_{i_{k+2}}^{a_{k+2,j}}\}$ to $t_j$ and the arcs of the complete
digraph on $\{x_{i_1}^{a_{1,j}}, x_{i_2}^{a_{2,j}},
\ldots{},x_{i_{k+2}}^{a_{k+2,j}}\}$.

\medskip

Let $D'_{{\cal F}}$ be the digraph obtained from $D_{{\cal F}}$ as
follows.  Add a set of $3m-n$ new vertices $U=\{u_1, \dots,
u_{3m-n}\}$ and let $f$ be a bijection between $U$ and $\{Y_i^p \mid
i\in [n], 1\leq p\leq m(i)-1\}$.  For each $j\in [3m-n]$, we add a
$(u_j,v_j)$-$(1,k+2)$-connector with $v_j$ being an arbitrary vertex
in $f(u_j)$, and a $(u_j, u_{j+1})$-$(k+1,k+2)$-connector (with
$u_{3m-n+1}=u_1$).  Finally, for each $j\in [m]$, add a
$(t_j,u_1)$-$(k+2,k+2)$-connector.  We can easily check that $D'_{\cal
  F}$ is strong and $(k+2)$-out-regular.

\medskip

Let us now prove that $D'_{\cal F}$ has a $(\Delta^+\leq
k,\Delta^+\leq k)$-partition if and only if $\cal F$ admits an
assignment such that each clause contains a true literal and a false
literal.  By Lemma~\ref{lem:connector}, as $k>0$, it is equivalent to
prove that $D_{\cal F}$ has a $(\Delta^+\leq k,\Delta^+\leq
k)$-partition if and only if $\cal F$ admits an assignment such that
each clause contains a true literal and a false literal.

First suppose that $\phi$ is a truth assignment such that
$\phi(x_i)\in \{true,false\}$ and each clause contains at least one
true and one false variable.  Define the following 2-colouring of
$V(D_{{\cal F}})$: for each $i\in [n]$ colour all vertices of
$\{x_i^{j} \mid j\in J(i)\}$ by colour 1 and those of
$\bigcup_{p=1}^{m(i)-1} Y_i^p$ by $2$ if $\phi(x_i)=true$ and
otherwise colour all vertices of $\{x_i^{j} \mid j\in J(i)\}$ by 2 and
those of $\bigcup_{p=1}^{m(i)-1} Y_i^p$ by $1$. Now each $t_j$, $j\in
[m]$, has at least one in-neighbour of colour $i$ for $i\in [2]$. If
it has precisely one of colour $i$, we colour it by colour $i$ and
otherwise we colour it arbitrarily. Now it is easy to see that letting
$V_i$ be the set of vertices of colour $i$, $i=1,2$, we obtain the
desired $2$-partition of $D_{\cal F}$.

Assume now that $(V_1,V_2)$ is a good $2$-partition of $D_{\cal F}$.
The forcing gadgets ensure that in every $(\Delta^+\leq k,\Delta^+\leq
k)$-partition $(V_1,V_2)$ of $V(D_{{\cal F}})$ all vertices of
$\{x_i^{j} \mid j\in J(i)\}$ belong to the same set in the partition
for all $i\in [n]$.  Furthermore, because of the complete subdigraphs
on the vertices
$\{x_{i_1}^{a_{1,j}},x_{i_2}^{a_{2,j}},\ldots{},x_{i_{k+2}}^{a_{k+2,j}}\}$,
$j\in [m]$, at least one of these vertices is in $V_1$ and at least
one of them is in $V_2$. Thus if we assign $x_i$ the value $true$ if
$\{x_i^{j} \mid j\in J(i)\}\subset V_1$ and $false$ otherwise, each
clause will have at least one true and at least one false literal.
\end{proof}

Combining our results above we obtain the following complete
classification in terms of $k_1,k_2$.

\begin{theorem}
Let $k_1, k_2$ be non-negative integers. The {\sc $(\Delta^+\leq
    k_1,\Delta^+\leq k_2)$-Partition} problem is
\begin{itemize}
\item polynomial-time solvable for all digraphs when $k_1=k_2=0$;
\item polynomial-time solvable for digraphs of maximum degree $p\leq \max\{k_1,k_2\}$;
\item ${\cal NP}$-complete for strong $p$-out-regular digraphs for all
  $p\geq \max\{k_1,k_2\}+1$ when $k_1\neq k_2$;
\item polynomial for $(k_2+1)$-out-regular digraphs and ${\cal
  NP}$-complete for strong $p$-out-regular digraphs for all $p\geq
  \max\{k_1,k_2\}+2$ when $k_1=k_2$.
\end{itemize}
\end{theorem}

Theorems \ref{degrees3} and \ref{thm:struc-k1-p1} this immediately
yield the following.
\begin{theorem}\label{thm:recap}
{\sc $k$-all-out-degree reducing $2$-partition} and {\sc
  $k$-max-out-degree-reducing $2$-partition} are polynomial-time
solvable for $k=1$ and ${\cal NP}$-complete for all integers $k\geq 2$
even when the input is a strong out-regular digraph.
\end{theorem}

\section{Out-degree reducing $p$-partitions for $p\geq 3$.} \label{sec:p-part}

All  our complexity results so far dealt with $2$-partition
problems.  In this section we deal with $p$-partitions for $p\geq 3$.

The next proposition implies that {\sc $k$-all-out-degree-reducing $p$-partition} and {\sc $k$-max-out-degree-reducing $p$-partition}
are polynomial-time solvable when $p\geq 2k+1$,
because the answer is trivially `yes'.

\begin{proposition}\label{prop:Alon-gen}
  Every digraph has a $k$-all-out-degree-reducing $(2k+1)$-partition
  and this is best possible.
\end{proposition}

\begin{proof}
For each vertex $v$ pick $\min\{k,d^+(v)\}$ arcs with tail in $v$.
Let $H$ be the subdigraph of $D$ induced by these arcs. Then $H$ has a
vertex of degree at most $2k$ and this holds for every subdigraph of
$H$, so $UG(H)$ is $2k$-degenerate and hence it is
$2k+1$-colourable. Let $(V_1,V_2,\ldots{},V_{2k+1})$ be a
$(2k+1)$-partition of $D$ induced by a $(2k+1)$-colouring of
$UG(H)$. It is easy to check that this is a
$k$-all-out-degree-reducing $(2k+1)$-partition since every arc of $H$
goes between two different sets in the partition.

The $k$-out-regular tournaments show that $2k+1$ is best possible for
each $k\geq 1$.  \end{proof}

The next result implies that {\sc $k$-all-out-degree-reducing $p$-partition} and {\sc $k$-max-out-degree-reducing $p$-partition}
 are also polynomial-time solvable when $p=2k$.

\begin{theorem}
Let $k \geq 2$.  A digraph $D$ admits a $k$-all-out-degree-reducing
$2k$-partition if and only if no terminal strong component of $D$ is a
$k$-regular tournament.
\end{theorem}

\begin{proof}
First assume that some terminal component, $Q$, of $D$ is a
$k$-regular tournament.  This implies that every vertex in $Q$ has
out-degree $k$ in $D$ and for any $2k$-partition of $D$ there will be
two vertices from $Q$ in the same part, as $|V(Q)|=2k+1$.
Therefore some vertex will have out-degree at least $1$ in its part and therefore not have reduced its out-degree by $k$. This proves
one direction. We now prove the opposite direction.

Let $D$ be any digraph of order $n$ and size $m$ with no terminal
component isomorphic to a $k$-regular tournament.  We will now show
that $D$ has a $k$-all-out-degree-reducing $2k$-partition by induction
on $n+m$. Clearly this holds when $n+m \leq 3$ so assume that it also
holds for all digraphs, $D'$, with $|V(D')|+|E(D')| < n+m$.  We may
assume that $D$ is connected as otherwise we are done by using
induction on each connected component.  Let
$G$ be the underlying graph of $D$. We consider the following
three cases which exhaust all possibilities.

\2

{\bf Case 1. There exists a $x \in V(D)$ with $d^+(x) > k$.}  If
$N^+(x)$ is independent then let $v \in N^+(x)$ be arbitrary, and
otherwise let $u,v \in N^+(x)$ be chosen such that $uv \in A(D)$. Let
$D' = D \setminus xv$ (i.e. delete the arc $xv$ from $D$).  Let $Q'$ be any
terminal component in $D'$. If $x \not\in V(Q')$, then $Q'$ is also a
terminal component of $D$ and therefore not a $k$-regular
tournament. So suppose $x \in V(Q')$.  Recall that either $N^+(x)$ is
independent or $xuv$ is a path in $D$ which implies that $v \in
V(Q')$. Both cases imply that $Q'$ is not a tournament.  Therefore, by
induction, there is a $k$-all-out-degree-reducing $2k$-partition of
$D'$ and therefore also of $D$ (using the same partition). This
completes Case 1.

\2

{\bf Case 2. $\Delta^+(D) \leq k$ and $G$ is not $2k$-regular.}  Let
$w$ be a vertex having degree at most $2k-1$ in $G$. Let $D' = D
-w$. Assume that some terminal component, $Q'$, in $D'$ is a
$k$-regular tournament.  As $\Delta^+(D) \leq k$, this implies that
$Q'$ is also a terminal component in $D$, a contradiction. Therefore
no terminal component in $D'$ is a $k$-regular tournament and by
induction there is a $k$-all-out-degree-reducing $2k$-partition of
$D'$.  Now add $w$ to a different part to all of its at most
$2k-1$ neighbours in $G$.  This gives a $k$-all-out-degree-reducing
$2k$-partition of $D$.

\2

{\bf Case 3. $\Delta^+(D) \leq k$ and $G$ is $2k$-regular.}  Note that in that case $D$ is an oriented graph and $D$ is $k$-regular. 
Now $G$ is not a complete graph for otherwise $D$ would be $k$-regular tournament. 
Moreover, as $k \geq 2$,  the graph $G$ is not an odd cycle. Therefore, by Brook's Theorem, $G$ admits a proper 
$2k$-colouring. This $2k$-colouring gives us the desired
$k$-all-out-degree-reducing $2k$-partition of $D$.
\end{proof}

\begin{theorem} \label{thm1}
 If $k>1$ and $3 \leq p \leq 2k-1$, then  {\sc $k$-all-out-degree-reducing $p$-partition} and {\sc $k$-max-out-degree-reducing $p$-partition} are $\cal NP$-complete.
\end{theorem}

\begin{proof}
 We give a reduction from {\sc $p$-Colourability} which consists in deciding whether a given digraph is $k$-colourable.
 This problem is well-known to be $\cal NP$-complete for all $p\geq 3$.

  We first need to define a gadget $D_2(x,y)$ as follows.
Let $T$ be a regular or almost regular tournament of order $p-1$ and let $V_1 = \{ v \, | \, d_T^+(v)=k-1 \}$.
Note that $V_1$ is empty if $p\leq 2k-2$ and $|V_1|=k-1=|V(T)|/2$ if $p=2k-1$.

Let $D_2(x,y)$ be the digraph obtained from a copy of $T$ by adding two vertices $x,y$ and all arcs from $V(T)\setminus V_1$ to $\{x,y\}$, all arcs from $V_1$ to $x$ and all arcs from $y$ to $V_1$.   
Note that $d^+(x) = 0$ and $d^+(y)=|V_1|$.  

 Note that in both cases above $x$ and $y$ are the only non-adjacent vertices in $D_2(x,y)$ and
$\Delta^+(D_2(x,y)) \leq k$. 

We now define the gadget $D_n(x_1,x_2,\ldots,x_n)$ for $n \geq 3$ as the union of $D_2(x_1,x_2)$, $D_2(x_2,x_3)$, ..., $D_2(x_{n-1},x_n)$,
where the copies of $T$ are disjoint.
Note that $d^+(x_1)=0$ and $d^+(x_i) \leq k-1$ for all $i=2,3,\ldots,n$ (in fact $d^+(x_i)=0$ if $p<2k-1$ and $d^+(x_i)=k-1$ otherwise).

We will now reduce an instance of {\sc $p$-Colourability} to an instance of {\sc $k$-max-out-degree-reducing $p$-partition}.
Let $G$ be a  graph with vertex set $v_1, \dots, v_n$. We will now construct a digraph $D$ as follows. 
For each vertex $v_i \in V(G)$ we let $D^i$ be a copy of $D_n(x^i_1,x^i_2,\ldots,x^i_n)$.
For each edge $v_iv_j$ of $G$ with $i<j$ add an arc from $x^i_j$ to $x^i_j$.
Observe that the set of arcs added by this operation are disjoint, so the resulting digraph $D$ has out-degree at most $k$.
Consequently, every $k$-max-out-degree-reducing $p$-partition  and every $k$-max-out-degree-reducing $p$-partition  of $D$
is equivalent to proper $p$-colouring of the underlying graph $UG(D)$ of $D$. 

Hence to prove the theorem, it is enough to show that $UG(D)$ has a proper $p$-colouring if and only if $G$ does. 
But this follows directly from the following claim.

\begin{claim}\label{claimA}
In any $p$-colouring of $UG(D_n(x_1,x_2,\ldots,x_n))$, all the vertices in $\{x_1,x_2,\ldots,x_n\}$ must be coloured the same.
Furthermore, there exists a $p$-colouring of $UG(D_n(x_1,x_2,\ldots,x_n))$.
\end{claim}

\noindent {\em Proof of Claim~\ref{claimA}.} We show Claim~\ref{claimA} is true when $n=2$ and then note that this implies that Claim~\ref{claimA} is true for all $n$. 
Let $n=2$.  As $x_1$ and $x_2$ are the only non-adjacent vertices in $D_2(x_1,x_2)$ and
$|V(D_2(x_1,x_2))|=p+1$ we note that $x_1$ and $x_2$ must have the same colour in a proper $p$-colouring of $UG(D_2(x_1,x_2))$.
Conversely if $x_1$ and $x_2$ have the same colour all other vertices of $D_2(x_1,x_2)$ can be given a distinct colour in order to 
obtain a proper $p$-colouring of the underlying graph.  This proves Claim~A when $n=2$. 

When $n \geq 3$ we note by the above that $x_1$ and $x_2$ must be in the same partite set. Analogously $x_2$ and $x_3$ must be in the same partite set. 
Continuing this process we obtain the desired result for $n \geq 3$. This completes the proof of Claim~\ref{claimA}.
\hfill $\Diamond$
\end{proof}

\section{Remarks and open questions}

A {\bf majority $k$-colouring} of a digraph $D=(V,A)$ is a
$k$-colouring of the vertices of $V$ so that each vertex $v$ has at
most $\frac{d^+(v)}{2}$ out-neighbours with the same colour as
itself. It is shown in \cite{kreutzerEJC24} that every digraph has a
majority 4-colouring and the authors conjecture that, in fact, every
digraph has a majority 3-colouring. They also asked about the
complexity of deciding whether a digraph has a majority
2-colouring. Since a 3-out-regular digraph has a majority 2-colouring
if and only if it has a $(\Delta^+\leq 1,\Delta^+\leq 1)$-partition
the following is an immediate consequence of Theorem \ref{degrees3}.

\begin{theorem}
Deciding whether a digraph has a majority $2$-colouring is ${\cal
  NP}$-complete even when the input is $3$-out-regular and strongly
connected.
\end{theorem}

\bigskip

In all our $\cal NP$-completeness proofs above on out-regular digraphs
these are far from being also in-regular. Thus it is natural to ask about
the complexity in the case of regular digraphs.

\begin{problem}
\label{k1k2regular}
What is the complexity of the $(\Delta^+\leq k_1,\Delta^+\leq
k_2)$-partition problem for $(\max\{k_1,k_2\}+1)$-regular digraphs
when $k_1<k_2$?
\end{problem}

\begin{problem}
\label{kkregular}
What is the complexity of the $(\Delta^+\leq k,\Delta^+\leq
k)$-partition problem for $(k+2)$-regular digraphs?
\end{problem}

Theorem \ref{thm:struc-k1-p1} implies that Problem \ref{kkregular}
becomes polynomial-time solvable if we replace $(k+2)$-regular by $(k+1)$-regular
and that when $k\geq 2$ a $(\Delta^+\leq k,\Delta^+\leq k)$-partition
always exists in every $(k+1)$-regular digraph as, by a result of
Thomassen \cite{thomassenJAMS5}, these all have an even directed cycle
(see also \cite[Theorem 8.3.7]{bang2009}).

\vspace{2mm}

Finally we can also ask about 2-partitions where the maximum
out-degree is reduced in one part whereas it is the maximum in-degree
that must be reduced in the other part.
\begin{problem}
What is the complexity of the $(\Delta^+\leq k_1,\Delta^-\leq
k_2)$-partition problem?
\end{problem}

\bigskip

In this paper, we studied partitions such that the out-degree in (the digraph induced by) each part is $k$ smaller than the out-degree in the whole digraph for some value $k$ which is fixed and the same for each part. It would be interesting to study the analogous problem where $k$ depends on the part.
In this vein Alon proved the following result.
\begin{theorem}[\cite{alonCPC15}]
\label{alonsplit}
Let $D$ be a digraph of maximum out-degree $\Delta^+$ and let
$d_1,d_2,\ldots{},d_p$ non-negative integers satisfying
$d_1+d_2+\ldots{}+d_p+(p-1)\geq 2\Delta^+$. Then $D$ has a
$p$-partition $(V_1,V_2,\ldots{},V_p)$ such that
$\Delta^+(\induce{D}{V_i})\leq d_i$.
\end{theorem}

\end{document}